%% file: VertexSparsifierCEdge.tex
\documentclass[11pt]{article}
\usepackage{amsmath,amsfonts,amsthm,amssymb,color}
\usepackage{fancybox}
\usepackage{hyperref}
\usepackage{xcolor}

\newtheorem{theorem}{Theorem}[section]
\newtheorem{corollary}[theorem]{Corollary}
\newtheorem{lemma}[theorem]{Lemma}

\theoremstyle{definition}
\newtheorem{definition}[theorem]{Definition}

\newenvironment{fminipage}%
  {\begin{Sbox}\begin{minipage}}%
  {\end{minipage}\end{Sbox}\fbox{\TheSbox}}

\newenvironment{algbox}[0]{\vskip 0.2in
\noindent 
\begin{fminipage}{6.3in}
}{
\end{fminipage}
\vskip 0.2in
}

\newcommand{\chat}{\widehat{c}}
\newcommand{\Ehat}{\widehat{E}}
\newcommand{\Fhat}{\widehat{F}}
\newcommand{\Ghat}{\widehat{G}}
\newcommand{\Shat}{\widehat{S}}
\newcommand{\Vhat}{\widehat{V}}
\newcommand{\Mincut}{\textsc{Mincut}}
\newcommand{\hc}{\hat{c}}
\renewcommand{\O}{\widetilde{O}}
\newcommand{\bs}{\backslash}
\newcommand{\localcut}{\mathrm{LocalCut}}

\newcommand{\cupdot}{\mathbin{\mathaccent\cdot\cup}}

\newcommand{\poly}{\mathrm{poly}}
\def\abs#1{\left|#1  \right|}

\newcommand{\vol}{\mathit{vol}}

\begin{document}

\title{Vertex Sparsifiers for $c$-Edge Connectivity}

\author{
	Yang P. Liu\\
	Stanford University\\
	\texttt{yangpliu@stanford.edu}
	\and
	Richard Peng\\
	Georgia Tech\footnote{Part of this work was done while visiting MSR Redmond}\\
	\texttt{richard.peng@gmail.com}
	\and
	Mark Sellke\\
	Stanford University\footnotemark[1]\\
	\texttt{msellke@stanford.edu}
}

\maketitle

\begin{abstract}
We show the existence of $O(f(c) \cdot k)$ sized vertex sparsifiers
that preserve all edge-connectivity values up to $c$ between
a set of $k$ terminal vertices, where $f(c)$ is a
function that only depends on $c$, the edge-connectivity value.
This construction is algorithmic: we also provide an algorithm whose
running time depends linearly on $k$, but exponentially in $c$.
It implies that for constant values of $c$,
an offline sequence of edge insertions/deletions and $c$-edge-connectivity
queries can be answered in polylog time per operation.
These results are obtained by combining structural results about
minimum terminal separating cuts in undirected graphs with recent
developments in expander decomposition based methods for finding
small vertex/edge cuts in graphs.
\end{abstract}

\section{Introduction}

Primitives that reduce the sizes of graphs while retaining key
properties are central to the design of efficient graphs algorithms.
Among such primitives, one of the most intriguing is the problem
of vertex sparsification: given a set of $k$ terminal vertices $S$,
reduce the number of non-terminal vertices while preserving key
information between the terminals.
This problem has been extensively studied in approximation
algorithms~\cite{MakarychevM10,CharikarLLM10,EnglertGKRTT14,
KrauthgamerR17,Kratsch12,AssadiKYV15,FafianieKQ16,FafianieHKQ16,
GoranciR16,GoranciHP17b}.
Recently, vertex sparsifiers were also shown to be
closely connected with dynamic graph
data structures~\cite{GoranciHP17a,PengSS19,GoranciHP18,DurfeeGGP19}.

Motivated by the problem of dynamic $c$-edge-connectivity,
which asks whether a pair of vertices have at least $c$ edge-disjoint
paths between them, we study vertex sparsifiers suitable these problems.
\begin{definition}
	\label{def:CLimitedVertexSparsifier}
	Two graphs $G$ and $H$ that both contain a subset of terminals
	$S$ are $(S, c)$-cut-equivalent if for any partition of $S$
	into $S = S_1 \cup S_2$, we have
	\[
	\min\left\{c, \min_{\substack{
			\Vhat \subseteq V\left( G \right)\\
			S_1 \subseteq \Vhat, 
			\Vhat \cap S_2 = \emptyset}}
		\abs{\partial_{G}\left( \Vhat \right)} \right\}
	=
	\min\left\{c, \min_{\substack{
			\Vhat \subseteq V\left( H \right)\\
			S_1 \subseteq \Vhat, 
			\Vhat \cap S_2 = \emptyset}}
	\abs{\partial_{H}\left( \Vhat \right)} \right\}.
	\]
	Here $\partial_G(\Vhat)$ denotes the set
	edges leaving $\Vhat$ in $G$, and is the same as
	$E(\Vhat, V \bs \Vhat)$.
\end{definition}
Our main result is that for any graph $G$ and terminals $S$, there is a
graph $H$ that is $(S,c)$-equivalent to $G$, where the size of $H$
depends linearly on the size of $S$, but exponentially on $c$. We call $H$
a $(S,c)$-\emph{vertex sparsifier} for $G$ on terminals $S$ as $H$ has far smaller size
than $G$ while maintaining the same $c$-edge connectivity information on the terminals $S$.

Furthermore, we utilize ideas from recent works on $c$-edge
connectivity~\cite{ForsterY19:arxiv,SaranurakW19,NanongkaiSY19b:arxiv}
to obtain efficient algorithms for the case where $c$ is a constant.
\begin{theorem}
	\label{thm:Main}
	Given any graph $G$ with $n$ vertices, $m$ edges,
	along with a subset of $k$ terminals $S$ and a value $c$,
	we can construct a graph $H$ which is
	$(S, c)$-cut-equivalent to $G$
	\begin{enumerate}
		\item \label{part:Main1} with $O(k \cdot O(c)^{c})$ edges
		in $O(m \cdot c^{O(c^2)} \cdot \log^{O(c)}n)$ time,
		\item \label{part:Main2} with $O(k \cdot O(c)^{2c})$ edges
		in $\O(m \cdot c^{O(c)})$ time \footnote{We use $\O$ to hide $\poly(\log n)$ factors. In particular, $\O(1)$ denotes $\poly(\log n)$.}.
	\end{enumerate}
\end{theorem}

Both components require algorithms for computing
expander decompositions (Lemma~\ref{lem:ExpanderDecompose}).
The first uses observations made in vertex cut
algorithms~\cite{NanongkaiSY19a,
NanongkaiSY19b:arxiv,ForsterY19:arxiv}, while the second
uses local cut algorithms developed from such studies.

The more general problem of multiplicatively preserving all edge connectivities
has been extensively studied.
Here an upper bound with multiplicative approximation factor
$O(\log{k} / \log\log{k})$~\cite{CharikarLLM10,MakarychevM10}
can be obtained without using additional vertices.
It's open whether this bound can be improved when additional
vertices are allowed, but without them,
a lower bound of $\Omega( (\log{k} / \log{\log{k}})^{1/4})$
is also known~\cite{MakarychevM10}.
For our restricted version of only considering values up to $c$,
the best existential bound for
larger values of $c$ is $O(k^3 c^2)$ vertices~\cite{Kratsch12,FafianieHKQ16}.
However, the construction time of these
vertex sparsifiers are also critical for their use in data structures~\cite{PengSS19}.
For a moderate number of terminals (e.g. $k = n^{0.1}$),
nearly-linear time constructions of vertex sparisifers with $\poly(k)$
vertices were known only when $c \leq 5$
previously~\cite{PengSS19, MolinaS18:unpublished}. 

Vertex sparsification is closely connected
with dynamic graph data structures, and directly plugging in these
sparsifiers as described in Theorem~\ref{thm:Main}
into the divide-and-conquer on time framework proposed by
Eppstein~\cite{Eppstein94} (a more general form of it
can be found in~\cite{PengSS19}) gives an efficient offline algorithm for
supporting fully dynamic $c$-connectivity queries.
\begin{theorem}[Dynamic offline connectivity]
An offline sequence of edge insertions, deletions, and $c$-connectivity queries on a $n$-vertex graph $G$ can be answered
in $\O(c^{O(c)})$ time per query.
\end{theorem}
In previously published works,
the study of fully dynamic connectivity has been limited to
the $c \leq 3$ setting~\cite{GalilI91,HolmLT98,LackiS13,KaczmarzL15,Kopeliovich12:thesis},

To our knowledge, the only results
for maintaining exact $c$ connectivity for $c \geq 4$
are an incremental algorithm
for $c = 4, 5$~\cite{DinitzV94,DinitzV95,DinitzW98},
and an unpublished offline fully dynamic algorithm
for $c = 4, 5$ by
Molina and Sandlund~\cite{MolinaS18:unpublished}. These algorithms all require about $\sqrt{n}$ time per query.

Furthermore, our algorithms gives a variety of connections
between graph algorithms, structural graph theory, and data structures:
\begin{enumerate}
\item The vertex sparsifiers we constructed can be viewed
as the analog of Schur complements (vertex sparsifiers for effective resistance)
for $c$-edge-connectivity, and raises the possibility that
algorithms motivated by Gaussian elimination~\cite{KyngLPSS16,KyngS16}
can in fact work for a much wider range of graph problems.
\item Our dependence on $c$ is highly suboptimal:
we were only able to construct instances that require
at least $2ck$ edges in the vertex sparsifier,
and are inclined to believe that an upper bound of $O(ck)$
is likely.
Narrowing this gap between upper and lower bounds is an interesting
question in combinatorial graph theory that directly affect the
performances of data structures and algorithms that utilize
such sparsifiers.
\item Finally, the recent line of work on turning offline Schur
complement based algorithms into online data structures~\cite{AbrahamDKKP16,DurfeeKPRS17,DurfeeGGP19}
suggest that our construction may also be useful
in online data structures for dynamic $c$-edge connectivity.
\end{enumerate}

\subsection{Paper Organization}
In Section \ref{sec:Prelim} we give preliminaries for our algorithm. In Section \ref{sec:Outline} we give an outline for our algorithms.
In Section \ref{sec:Existence} we show the existence of good $(S,c)$-cut vertex sparsifiers. In Section \ref{sec:Local} we give a polynomial time
construction of $(S,c)$-cut vertex sparsifiers whose size is slightly larger than those given in Section \ref{sec:Existence}. In Section \ref{sec:Expander} we use expander decomposition to make our algorithms run in nearly linear time.

\section{Preliminaries}
\label{sec:Prelim}

\subsection{General Notation}
All graphs that we work with are undirected and unit-weighted,
but our treatement of cuts and contractions naturally
require (and lead to) multi-edges.
We will refer to cuts as both subsets of edges, $F \subseteq E$,
or the boundary of a subset of vertices
\[
\partial\left(\Vhat\right)
=
E\left( \Vhat, V \bs \Vhat\right)
=
\left\{ e= (u, v) \in E \mid u \in \Vhat, v \notin \Vhat \right\}.
\]
For symmetry, we will also denote cuts using the notation
$(V_1, V_2)$, with $V = V_1 \cupdot V_2$,
where $\cupdot$ is disjoint union.

We will use $S \subseteq V$ to denote a subset of terminals, and define $k = |S|.$
Note that each cut $(V_1, V_2)$ naturally induces a partition
of the terminals into $S_1 = S \cap V_1$
and $S_2 = S \cap V_2$.
For the reverse direction, given
two subsets of terminals $S_1, S_2 \subseteq V$,
we use $\Mincut(G, S_1, S_2)$ to denote an arbitrary
minimum cut between $S_1$ and $S_2$, and we let $|\Mincut(G, S_1, S_2)|$ denote its size.
Note that if $S_1$ and $S_2$ overlap, this value is
infinite: such case does not affect
Definition~\ref{def:CLimitedVertexSparsifier} because
it naturally takes the minimum with $c$.
On the other hand, it leads us to focus more on disjoint
splits of $S$, and we denote such splits using
$S = S_1 \cupdot S_2$.
For a set of terminals $S$, we refer to the set of cuts
separating them with size at most $c$ as the
$(S, c)$-cuts.

A \emph{terminal cut} is any cut that has at least one terminal on both sides of the cut. We also
sometimes refer to these as \emph{Steiner cuts}, as this language has been used in the past work
of Cole and Hariharan \cite{ColeH03}. The minimum terminal cut or minimum Steiner cut is the terminal cut
with the smallest number of edges.

\subsection{Contractions and Cut Monotonicity}

Our algorithm will use the concept of contractions. For a graph $G$ and edge $e \in E(G)$,
we let $G/e$ denote the graph with the endpoints of $e$
identified as a single vertex, and we say that we have \emph{contracted} edge $e$.
The new vertex is made a terminal if at least
one of the endpoints was a terminal. For any subset $F \subseteq E$, we let $G/F$
denote the graph where all edges in $F$ are contracted.
We can show that for any split of terminals, the value
of the min-cut between them is monotonically increasing under such contractions.
\begin{lemma}
	\label{lem:CutMonotone}
	For any split of termainls $S = S_1 \cup S_2$,
	and any set of edges $F$, we have
	\[
	\left| \Mincut\left(G, S_1, S_2\right) \right|
	\leq
	\left| \Mincut\left(G/F, S_1, S_2\right) \right|.
	\]
\end{lemma}

\subsection{Observations about $(S,c)$-cut equivalence}

We start with several observations about the notion
of $(S,c)$-cut-equivalence given in Definition~\ref{def:CLimitedVertexSparsifier}.

\begin{lemma}
\label{lem:Smaller}
	If $G$ and $H$ are $(S,c)$-cut equivalent, then
	for any subset $\Shat$ of $S$,
	and any $\chat \leq c$, $G$ and $H$ are also
	$(\Shat, \chat)$-equivalent
\end{lemma}
This notion is also robust to the addition of edges.
\begin{lemma}
\label{lem:AddEdges}
	If $G$ and $H$ are $(S, c)$-cut-equivalent, then for
	any additional set of edges with endpoints in $S$,
	$G \cup F$ and $H \cup F$
	are also $(S, c)$-cut-equivalent.
\end{lemma}

When used in the reverse direction, this lemma says that we
can remove edges, as long as we include their endpoints as
terminal vertices.
\begin{corollary}
	\label{lem:RemoveEdges}
	Let $F$ be a set of edges in $G$ with endpoints
	$V(F)$, and $S$ be a set of terminals in $G$.
	If $H$ is $(S \cup V(F), c)$-cut equivalent
	to $G \bs F$, then $H \cup F$, which is
	$H$ with $F$ added, is $(S, c)$-cut equivalent to $G$.
\end{corollary}
We complement this partition process by showing that sparsifiers
on disconnected graphs can be built separately.
\begin{lemma}
\label{lem:Independent}
	If $G_1$ is $(S_1, c)$-cut-equivalent to $H_1$,
	and $G_2$ is $(S_2, c)$-cut-equivalent to $H_2$,
	then the vertex-disjoint union
	of $G_1$ and $G_2$,
	is $(S_1 \cup S_2, c)$-cut-equivalent
	to the vertex-disjoint union of $H_1$ and $H_2$.
\end{lemma}
When $F \subseteq E$ is a cut, combining Lemma~\ref{lem:RemoveEdges} and \ref{lem:Independent} allow us to recurse
on the connected components of $G\bs F$, provided that we
add the endpoints of the edges in $F$ as terminals.

\subsection{Edge Reductions}

Furthermore, we can restrict our attention to sparse graphs
only~\cite{NagamochiI92}.
\begin{lemma}
	\label{lem:EdgeReduction}
	Given any graph $G = (V, E)$ on $n$ vertices and any $c \geq 0$,
	we can find in $O(cm)$ time a graph $H$ on the same
	$n$ vertices, but with at most $c(n - 1)$ edges,
	such that $G$ and $H$ are $(S, c)$-cut-equivalent.
\end{lemma}

\begin{proof}
	Consider the following routine: repeat $c$ iterations of
	finding a maximal
	spanning forest from $G$, remove it from $G$ and add it to $H$.
	
	Each of the steps takes $O(m)$ time, for a total of $O(mc)$.
	Also, a maximal spanning tree has the property that for non-empty
	cut, it contains at least one edge from it.
	Thus, for any cut $\partial(S)$, the $c$ iterations adds at least
	\[
	\min\left\{c, \abs{\partial\left( S \right)} \right\}
	\]
	edges to $H$, which means up to a value of $c$, all cuts
	in $G$ and $H$ are the same.
\end{proof}
Note however that sparse is not the same as bounded degree:
for a star graph, we cannot reduce the degree of its center vertex
without changing connectivity.

\subsection{Edge Containment of Terminal Cuts}

Our construction of vertex sparsifiers utilizes an intermediate
goal similar to the construction of $(S, 5)$-cut-sparsifiers
by Molina and Sandlund~\cite{MolinaS18:unpublished}.
Specifically, we want to find a subset of edges $F$ so that for any
separation of $S$ has a minimum cut using only the edges from $F$.

\begin{definition}
	\label{def:Contain}
	In a graph $G = (V, E)$ with terminals $S$,
	a subset of edges $F \subseteq E$ is said to contain all
	$(S, c)$-cuts if for any split
	$S = S_1 \cupdot S_2$ with $\hc = |\Mincut(G,S_1,S_2)| \le c$,
	there is a subset of $F$ of size $\hc$ which
	is a cut between $S_1$ and $S_2$.
\end{definition}

Note that this is different than containing all the minimum cuts:
on a length $n$ path with two endpoints as terminals,
any intermediate edge contains a minimum terminal cut,
but there are up to $n -1$ different such minimum cuts.

Such containment sets are useful because
we can form a vertex sparsifier by
contracting the rest of the edges.
\begin{lemma}
	\label{lem:EdgesToSparsifier}
	If $G = (V, E)$ is a connected graph with terminals $S$, and
	$F$ is a subset of edges that contain all $(S, c)$-cuts,
	then the graph
	\[
	H = G / \left(E \bs F\right)
	\]
	is a $(S, c)$-cut-equivalent to $G$,
	and has at most $|F| + 1$ vertices.
\end{lemma}

\begin{proof}
	Consider any cut using entirely edges in $F$:
	contracting edges from $E \bs F$ will bring together
	vertices on the same side of the cut.
	Therefore, the separation of vertices given by this cut
	also exists in $H$ as well.

	To bound the size of $H$, observe that contracting all edges
	of $G$ brings it to a single vertex.
	That is, $H / F$ is a single vertex: uncontracting an edge
	can increase the number of vertices by at most $1$,
	so $H$ has at most $|F| + 1$ vertices.
\end{proof}

We can also state Lemmas \ref{lem:RemoveEdges} and \ref{lem:Independent} in the
language of edge containment.

\begin{lemma}
\label{lem:ContainRemoveEdges}
	Let $F$ be a set of edges in $G$ with endpoints
	$V(F)$, and $S$ be a set of terminals in $G$.
	If edges $\hat{F}$ contain all $(S \cup V(F), c)$-cuts
	in $G \bs F$, then $\hat{F} \cup F$
	contains all $(S, c)$-cuts in $G$.
\end{lemma}
\begin{lemma}
\label{lem:ContainIndependent}
	If the edges $F_1 \subseteq E(G_1)$ contain all $(S_1,c)$-cuts in $G_1$, and
	the edges $F_2 \subseteq E(G_2)$ contain all $(S_2,c)$-cuts in $G_2$
	, then $F_1 \cup F_2$ contains
	all the $(S_1 \cup S_2, c)$-cuts in the vertex disjoint union of $G_1$ and $G_2$.
\end{lemma}

\section{High Level Outline}
\label{sec:Outline}

Our construction is based on repeatedly finding edges that intersect
all $(S,c)$-cuts.

\begin{definition}
	\label{def:Intersect}
	In a graph $G = (V, E)$ with terminals $S$,
	a subset of edges $F \subseteq E$
	intersects all $(S, c)$-cuts for some $c > 0$
	if for any split $S = S_1 \cupdot S_2$ with
	$\hc = |\Mincut(G,S_1,S_2)| \le c$, there exists
	a cut $\Fhat = E(V_1, V_2)$ such that:
	\begin{enumerate}
		\item $\Fhat$ has size $\hc$,
		\item $\Fhat$ induces the same separation of $S$:
		$V_1 \cap S = S_1$, $V_2 \cap S = S_2$.
		\item $\Fhat$ contains at most $c - 1$ edges from
		any connected component of $G \bs F$.
	\end{enumerate}
\end{definition}

We can reduce the problem of finding edges that contain all small cuts
to the problem finding edges that intersect all small cuts.
This is done by first finding an intersecting set $F$,
and then repeating on the (disconnected) graph with $F$ removed,
but with the endpoints of $F$ included as terminals as well.

\begin{lemma}
	\label{lem:IntersectToContain}
	If in some graph $G = (V, E)$ with terminals $S$,
	a subset of edges $F \subseteq E$ intersects all
	$(S, c)$-cuts, then consider the set
	\[
	\Shat = S \cup V\left( F \right),
	\]
	that is, $S$ with the endpoints of $F$ added.
	If a subset $\Fhat$ contains all
	$(\Shat, c-1)$ cuts in the graph $(V, E \bs F)$,
	then $F \cup \Fhat$ contains all $(S, c)$-cuts in $G$
	as well.
\end{lemma}

\begin{proof}
Consider a partition $S = S_1 \cupdot S_2$ with $\hat{c} = |\Mincut(G,S_1,S_2)| \le c$. Because $F$
intersects all $(S,c)$-cuts, there is cut of size $\hat{c}$ separating $S_1$ and $S_2$ that has
at most $c-1$ edges in each connected component of $G \bs F$.

Combining this with Lemmas \ref{lem:ContainRemoveEdges} and \ref{lem:ContainIndependent} shows that
if $\Fhat$ contains all $(S\cup V\left(F\right),c-1)$-cuts in $G\bs F$, then $F \cup \Fhat$ contains all $(S,c)$
in $G$.
\end{proof}

Sections \ref{sec:Existence}, \ref{sec:Local}, and \ref{sec:Expander} are devoted to showing the following bound for generating a
sets of edges that intersect all $(S, c)$-cuts.
\begin{theorem}
	\label{thm:IntersectingEdges}
	For any parameter $\phi$ and any value $c$,
	for any graph $G$ with terminals $S$, we can generate
	a set of edges $F$ that intersects all
	$(S, c)$-cuts:
	\begin{enumerate}
		\item 	\label{part:Slower}
		with size at most $O((\phi m \log^4{n} + |S|) \cdot c )$
		in $\O(m(c\phi^{-1})^{2c})$ time.

		\item 	\label{part:Faster}
		with size at most $O((\phi m \log^{4}n + |S|) \cdot c^2)$
		in $\O(m \phi^{-2} c^7)$ time.
	\end{enumerate}
\end{theorem}

Then the overall algorithm is simply to iterate this process
until $c$ reaches $1$, as in done in Figure \ref{fig:GetContainingEdges}.

\begin{figure}
	\begin{algbox}
		$F = \textsc{GetContainingEdges}(G, S, c)$
		
		\underline{Input}: undirected unweighted multi-graph $G$,
		terminals $S$,
		cut threshold $c$.
		
		\underline{Output:} set of edges $F$ that contain all interesting
		$(S,c)$-cuts.
		\begin{enumerate}
			\item Initialize $F \leftarrow \emptyset$.
			\item For $\chat \leftarrow c \ldots 1$ in decreasing order: \label{line:chat}
			\begin{enumerate}
				\item $F \leftarrow F \cup \textsc{GetIntersectingEdges}(G, S, \chat)$.
				\item $G \leftarrow G \bs F$.
				\item $S \leftarrow S \cup V(F)$, where $V(F)$ is the endpoints of all edges of $F$.
			\end{enumerate}
			\item Return $F$.
		\end{enumerate}
	\end{algbox}
	\caption{Pseudocode For Finding a Set of Edges
		that contain all $(S,c)$-cuts}
	\label{fig:GetContainingEdges}
\end{figure}

\begin{proof}[Proof of Theorem~\ref{thm:Main}]
We first show part \ref{part:Main1} of Theorem \ref{thm:Main}.

\newcommand{\cintersect}{C_{\mathit{intersect}}}

Let $\cintersect$ be a constant such that part \ref{part:Slower} of Theorem \ref{thm:IntersectingEdges} gives us
a set $F$ of edges intersecting all $(S,c)$-cuts of size at most $A(\phi m \log^4{n} + |S|)c$ in $\O(m(c\phi^{-1})^c)$ time.
We show by induction that before processing $\chat = i$ in line \ref{line:chat} of Figure \ref{fig:GetContainingEdges} that
\[
\left|V(F)\right|
\le
\left(4 \cintersect\right)^{c-i}\frac{c!}{i!}
\left(\phi m \log^4 n + |S|\right)
\]
and
\[
\left|F\right|
\le
\left(4 \cintersect\right)^{c-i}
\frac{c!}{i!}
\left(\phi m \log^4 n + \left|S\right|\right).
\]

We focus on the bound on $|V(F)|$, as the bound on $|F|$ is similar. The induction hypothesis holds for $i = c$.
By Part~\ref{part:Slower} of Theorem~\ref{thm:IntersectingEdges}
we have the size of $F$ after processing $\chat = i$ is at most
\begin{align*}
\cintersect
\left(\phi m \log^4{n} + |V(F)|\right)c
&\le
\cintersect
\left(\phi m \log^4{n} + \left(4A\right)^{c-i}\frac{c!}{i!}
\left(\phi m \log^4 n + \left|S\right| \right)\right)i \\
&\le \frac{1}{2}
\left(4 \cintersect\right)^{c-i+1}
\frac{c!}{\left(i-1\right)!}
\left(\phi m \log^4 n + \left|S\right|\right).
\end{align*}
Now the size of $V(F)$ is at most twice this bound, as desired. Taking $i = 0$ shows that the final size of $F$ is at most
$(4 \cintersect)^c c!\left(\phi m \log^4 n + |S|\right).$
Take $\phi = \frac{1}{5c(4 \cintersect)^c c!\log^4n}.$
For $m \le nc$ (which we can assume by Lemma \ref{lem:EdgeReduction})
the final size of $F$ is at most
\[
\left(4 \cintersect\right)^c c!\left(\phi m \log^4 n + |S|\right)
\le \frac{n}{5} + \left(4 \cintersect \right)^c c! \left|S\right|.
\]
Now, we apply Lemma \ref{lem:EdgesToSparsifier} to produce a graph $H$ with at most $\frac{n}{5} + (4\cintersect)^c c! |S| + 1$ vertices that is $(S,c)$-cut equivalent to $G$.
Now, we can repeat the process on $H$ $O(\log n)$ times.
The number of vertices in the graphs we process decrease geometrically until they have at most $2(4 \cintersect)^c c! |S| = O(c)^c |S|$ many vertices,
as desired.

The runtime bound in part~\ref{part:Main1} of Theorem~\ref{thm:Main}
follows from the runtime bound in Part~\ref{part:Slower} of
Theorem~\ref{thm:IntersectingEdges} and our choice of $\phi$.

The analysis of part \ref{part:Main1} of Theorem \ref{thm:Main} follows from part \ref{part:Faster} of Theorem \ref{thm:IntersectingEdges} in a similar way.
\end{proof}

\section{Existence via Structural Theorem and Recursion}
\label{sec:Existence}

Our algorithm is based on a divide-and-conquer routine
that removes a small cut and recurses on both sides.
Our divide-and-conquer relies on the following observation
about when $(S, c)$-cuts are able to interact
completely with both sides of a cut.

\begin{lemma}
	\label{lem:Partition}
	Let $F$ be a cut given by the partition $V = V_1 \cupdot V_2$
	in $G = (V, E)$ such that both $G[V_1]$	and $G[V_2]$ are connected,
	and $S_1 = V_1 \cap S$ and $S_2 = V_2 \cap S$ be the partition
	of $S$ induced by this cut.
	If $F_1$ intersects all $(S_1 \cup \{v_2\}, c)$-terminal
	cuts in $G / V_2$, the graph formed by contracting all of
	$V_2$ into a single vertex $v_2$, and similarly $F_2$
	intersects all $(S_2 \cup \{v_1\}, c)$-terminal
	cuts in $G / V_1$, then $F_1 \cup F_2 \cup F$ intersects
	all $(S, c)$-cuts in $G$ as well.
\end{lemma}

\begin{proof}
	Consider some cut $\Fhat$ of size at most $c$.
	
	If $\Fhat$ uses an edge from $F$,
	then it has at most $c - 1$ edges in $G \bs F$,
	and thus in any connected component as well.
	
	If $\Fhat$ has at most $c - 1$ edges in $G[V_1]$, then
	because removing $F$ already disconnected $V_1$ and $V_2$,
	and removing $F_1$ can only further disconnect things,
	no connected component in $V_1$ can have $c$ or more edges.
	
	So the only remaining case is if $\Fhat$ is entirely
	contained on one of the sides.
	Without loss of generality assume $\Fhat$ is entirely
	contained in $V_1$,
	i.e. $\Fhat \subseteq E(G[V_1])$.
	Because no edges from $G[V_2]$ are removed and $G[V_2]$ is connected,
	all of $S_2$ must be on one side of the cut,
	and can therefore be represented by a single vertex $v_2$.
	
	So using the induction hypothesis on the cut $\Fhat$
	in $G / V_2$ with the terminal separation given by
	all of $S_2$ replaced by $v_2$ gives that $\Fhat$
	has at most $c - 1$ edges in any connected component of
	\[
	\left( G / V_2 \right) \bs F_1.
	\]
	Because connected components are unchanged under contracting
	connected subsets, we get that $\Fhat$ has at most $c - 1$
	edges in any connected components of $G \bs F_1$ as well.
\end{proof}

However, for such a partition to make progress,
we also need at least two terminals to become contracted
together when $V_1$ or $V_2$ are contracted.
Building this into the definition leads to our key definition
of a non-trivial $S$-separating cut:
\begin{definition}
\label{def:Nontrivial}
	A non-trivial $S$-separating cut is a separation
	of $V$ into $V_1 \cupdot V_2$ such that:
	\begin{enumerate}
		\item the induced subgraphs on $V_1$ and $V_2$,
		$G[V_1]$ and $G[V_2]$ are both connected.
		\item $|V_1 \cap S| \geq 2$, $|V_2 \cap S| \geq 2$.
	\end{enumerate}
\end{definition}

Such cuts are critical for partitioning and recursing on
the two resulting pieces. Connectivity of $G[V_1]$ and $G[V_2]$ is necessary for
applying Lemma \ref{lem:Partition}, and $|V_1 \cap S| \geq 2$, $|V_2 \cap S| \geq 2$
are necessary to ensure that making this cut and recursing makes progress.

We now study the set of graphs $G$ and terminals $S$ for which a nontrivial cut exists.
For example, consider for example when $G$ is a star graph (a single vertex with $n-1$ vertex connected to it)
and all vertices are terminals.
In this graph, the side of the cut not containing the center
can only have a single vertex, hence there are no nontrivial cuts.

We can, in fact, prove the converse:
if no such interesting separations exist,
we can terminate by only considering the $|S|$ separations of $S$
formed with one terminal on one of the sides.
We define these cuts to be the $s$-isolating cuts.

\begin{definition}
	\label{def:IsolatingCut}
	For a graph $G$ with terminal set $S$ and some $s \in S$,
	a $s$-isolating cut is a split of the vertices
	$V = V_{A} \cupdot V_{B}$
	such that $s$ is the only terminal in $V_A$, i.e.
	$s \in V_{A}$, $(S \bs \{u\}) \subseteq V_{B}$.
\end{definition}

\begin{lemma}
	\label{lem:OPLemma}
	If $S$ is a subset of at least $4$
	terminals in an undirected graph $G$
	such that there does not exist a non-trivial $S$ separating cut
	of size at most $c$, then
	\[
		F =
		\bigcup_{\substack{s \in S \\ \left|
			\Mincut\left(G,\left\{s\right\}, S \bs \left\{s\right\}\right) \right| \leq c}}
		 \Mincut\left(G,\left\{s\right\}, S \bs \left\{s\right\}\right) 
	\]
	contains all $(S, c)$ cuts of $G$. Here, $F$ is the union of all $s$-isolating cuts of size at most $c$.
\end{lemma}

\begin{proof}
	Consider a graph with no non-trivial $S$ separating cut
	of size at most $c$, but there is a partition of $S$,
	$S = S_1 \cupdot S_2$,
	such that the minimum cut between $S_1$ and $S_2$,
	$V_1$ and $V_2$, has at most $c$ edges, and $|S_1|, |S_2| \ge 2.$
	
	Let $\Ehat$ be one such cut,
	and consider the graph
	\[
	\Ghat = G / \left(E \bs \Ehat\right),
	\]
	that is, we contract all edges except the ones on this cut. Note that $\Ghat$ has at least $2$ vertices.

	Consider a spanning tree $T$ of $\Ghat$.
	By minimality of $\Ehat$, each node of $T$ must contain
	at least one terminal. Otherwise, we can keep one edge
	from such a node without affecting the distribution of terminal vertices.
	
	We now show that no vertex of $T$ can contain $|S|-1$ terminals.
	If $T$ has exactly two vertices, then one vertex must correspond to $S_1$
	and one must correspond to $S_2$, so no vertex has $|S|-1$ terminals. If $T$ has
	at least $3$ vertices, then because every vertex contains at least one terminal,
	no vertex in $T$ can contain $|S|-1$ vertices. 

	Also, each leaf of $T$ can contain at most one terminal, otherwise
	deleting the edge adjacent to that leaf forms a nontrivial cut.

	Now consider any non-leaf node of the tree, $r$.
	As $r$ is a non-leaf node, at it has at least two different
	neighbors that lead to leaf vertices.

	Reroot this tree at $r$, and consider some neighbor of $r$, $x$.
	If the subtree rooted at $x$ has more than $2$ terminals,
	then cutting the $rx$ edge results in two components,
	each containing at least two terminals
	(the component including $r$ has at least one other neighbor
	that contains a terminal).
	Thus, the subtree rooted at $x$ can contain at most one
	terminal, and must therefore be a singleton leaf.
	
	Hence, the only possible structure of $T$ is a star
	centered at $r$ (which may contain multiple terminals),
	and each leaf having a exactly one terminal in it.
	This in turn implies that $\Ghat$ also must be a star,
	i.e. $\Ghat$ has the same edges as $T$ but possibly with multi-edges.
	This is because any edge between two leaves of a star forms
	a connected cut by disconnecting those vertices from $r$.

	By minimality, each cut separating the root from leaf is a
	minimal cut for that single terminal, and these cuts
	are disjoint.
	Thus taking the union of edges of all these singleton cuts
	gives a cut that splits $S$ the same way,
	and has the same size.	
\end{proof}

Note that Lemma \ref{lem:OPLemma} is not claiming all the $(S,c)$-cuts
of $S$ are singletons. Instead, it says that any $(S,c)$-cut can be formed
from a union of single terminal cuts.

Combining Lemma \ref{lem:Partition} and \ref{lem:OPLemma}, we obtain the recursive algorithm in 
Figure~\ref{fig:RecursiveNonTrivialCuts}, which demonstrates
the existence of $O(|S| \cdot c)$ sized
$(S, c)$-cut-intersecting subsets. If there is a nontrivial $S$-separating cut,
the algorithm in Line~\ref{ln:Partition} finds it and recurses on both sides of the cut using Lemma \ref{lem:Partition}.
Otherwise, by Lemma \ref{lem:OPLemma}, the union of the $s$-isolating cuts of size at most $c$
contains all $(S,c)$-cuts, so the algorithm keeps the edges of those cuts in Line~\ref{ln:OP}.

\begin{figure}[!h]
		\begin{algbox}
		$F = \textsc{RecursiveNontrivialCuts}(G, S, c)$
		
		\underline{Input}: undirected unweighted multi-graph $G$,
		terminals $S$,
		cut threshold $c > 0$.
		
		\underline{Output:} set of edges $F$ that intersect all $S$
		separating cuts of size at most $c$
		
		\begin{enumerate}
			\item If $|S| \leq 4$, return union of the
			min-cuts of all $2^{|S| - 1} \leq 8$ splits of the terminals.
			\item Initialize $F \leftarrow \emptyset$.
			\item \label{ln:Partition} If there exists some non-trivial
			$S$-separating cut $(V_1, V_2)$ of size at most $c$,
			\begin{enumerate}
				\item $F \leftarrow F \cup E(V_1, V_2)$.
				\item $F \leftarrow F \cup \textsc{RecursiveNontrivialCuts}(G / V_2, (S \cap V_1) \cup \{v_2\}, c)$ where $v_2$ is the vertex that $V_2$
				gets contracted to in $G / V_2$.
				\item $F \leftarrow F \cup \textsc{RecursiveNontrivialCuts}(G / V_1, (S \cap V_2) \cup \{v_1\}, c)$ where $v_1$ is the vertex that $V_1$
				gets contracted to in $G / V_1$.
			\end{enumerate}
			\item Else add all local terminal cuts to $F$:
			\begin{enumerate}
				\item \label{ln:OP} For all vertex $v$ such that $|\Mincut(G, v, S \bs v)| \leq c$,
				do
				\[
				F \leftarrow F \cup
				\Mincut(G, v, S \bs v).
				\]
			\end{enumerate}
			\item Return $F$.
		\end{enumerate}
	\end{algbox}
	\caption{Algorithm for finding a set of edges
		that intersects all terminal cuts of size $\leq c$.}
	\label{fig:RecursiveNonTrivialCuts}
\end{figure}

\begin{lemma}
	\label{lem:RecursiveNontrivialCutsCorrectness}
	\textsc{RecursiveNonTrivialCuts} as shown in 
	Figure~\ref{fig:RecursiveNonTrivialCuts} correctly returns
	a set of $(S, c)$-cut-intersecting edges of size at most
	$O(|S| \cdot c)$.
\end{lemma}

\begin{proof}
	Correctness can be argued by induction.
	The base case of where we terminate
	by adding all min-cuts with one terminal on one side,
	follows from Lemma~\ref{lem:OPLemma}, while the inductive
	case follows from applying Lemma~\ref{lem:Partition}.
	
	It remains to bound the size of $F$ returned.
	Once again there are two cases:
	for the case where we terminate with the union 
	of singleton cuts, each such cut has size at most $c$,
	for a total of $|S| \cdot c$.
	
	For the recursive case, the recursion can be viewed as
	splitting $k \geq 4$ terminals into two instances of
	sizes $k_1$ and $k_2$ where $k_1 + k_2 = k + 2$
	and $k_1, k_2 \geq 3$.
	Note that the total values of $|\textsc{Terminals}| - 2$
	across all the recursion instances is strictly decreasing,
	and is always positive.
	So the recursion can branch at most $|S|$ times,
	which gives that the total number of edges added is at most
	$O(c \cdot |S|)$.
\end{proof}

\section{Poly-time Construction}
\label{sec:Local}

While the previous algorithm in Section~\ref{sec:Existence}
gives our best bound on sparsifier size,
it is not clear to us how
it could be implemented in polynomial time.
While we do give a more efficient implementation of it below
in Section~\ref{sec:Expander}, the running time of that algorithm
still has a $\log^{O(c)}{n}$ term
(as stated in Theorem~\ref{thm:Main}~Part~\ref{part:Slower}).
In this section, we give a more efficient algorithm
that returns sparsifiers of larger size,
but ultimately leads to the faster running time given
in Theorem~\ref{thm:Main}~Part~\ref{part:Faster}.
It was derived by working backwards from the termination condition
of taking all the cuts with one terminal on one side
in Lemma~\ref{lem:OPLemma}.

Recall that a Steiner cut is a cut with at least one terminal one both sides.
The algorithm has the same high level recursive structure,
but it instead only finds the minimum Steiner cut or certifies that its size is greater than $c$.
This takes $\O(m+nc^3)$ time using
an algorithm by Cole and Hariharan~\cite{ColeH03}.

It is direct to check that both sides of a minimum Steiner cut are connected.
This is important towards our goal of finding a non-trivial $S$-separating cut, defined in Definition \ref{def:Nontrivial}.

\begin{lemma}
	If $(V_A, V_B)$ is the global minimum $S$-separating cut in
	a connceted graph $G$, then both $G[V_A]$ and $G[V_B]$ must
	be connected.
\end{lemma}

\begin{proof}
	Suppose for the sake of contradiction that $V_{A}$ is
	disconnected.
	That is, $V_{A} = V_{A1} \cupdot V_{A2}$, there
	are no edges between $V_{A1}$ and $V_{A2}$.
	
	Without loss of generality assume $V_{A1}$ contains
	a terminal.
	Also, $V_{B}$ contains at least one terminal because
	$(V_A, V_B)$ is $S$-separating.
	
	Then because $G$ is connected, there is an edge between
	$V_{A1}$ and $V_{B}$.
	Then the cut $(V_{A1}, V_{A2} \cup V_{B})$ has strictly fewer
	edges crossing, and also terminals on both sides, a contradiction
	to $(V_A, V_B)$ being the minimum $S$ separating cut. 
\end{proof}

So the only bad case that prevents us from recursing is
the case where the minimum Steiner cut has a single terminal $s$
on some side. That is, one of the $s$-isolating cuts from
Definition~\ref{def:IsolatingCut} is also a minimum Steiner cut.

We can handle this case through an extension of Lemma \ref{lem:Partition}.
Specifically, we show that for a cut with both sides connected, we can
contract a side of the cut along with the cut edges before recursing.
\begin{lemma}
	\label{lem:IsolateContract}
	Let $F$ be a cut given by the partition $V = V_1 \cupdot V_2$
	in $G = (V, E)$ such that both $G[V_1]$	and $G[V_2]$ are connected,
	and $S_1 = V_1 \cap S$ and $S_2 = V_2 \cap S$ be the partition
	of $S$ induced by this cut.
	If $F_1$ intersects all $(S_1 \cup \{v_2\}, c)$-terminal
	cuts in $G/V_2/F$, the graph formed by contracting all of
	$V_2$ and all edges in $F$ into a single vertex $v_2$, and similarly $F_2$
	intersects all $(S_2 \cup \{v_1\}, c)$-terminal
	cuts in $G/V_1/F$, then $F_1 \cup F_2 \cup F$ intersects
	all $(S, c)$-cuts in $G$ as well.
\end{lemma}

\begin{proof}
	Consider some cut $\Fhat$ of size at most $c$.
	
	If $\Fhat$ uses an edge from $F$,
	then it has at most $c - 1$ edges in $G \bs F$,
	and thus in any connected component as well.
	
	If $\Fhat$ has at most $c - 1$ edges in $G[V_1]$, then
	because removing $F$ already disconnected $V_1$ and $V_2$,
	and removing $F_1$ can only further disconnect things,
	no connected component in $V_1$ can have $c$ or more edges.
	
	The only remaining case is if $\Fhat$ is entirely
	contained on one of the sides.
	Without loss of generality assume $\Fhat$ is entirely
	contained in $V_1$,
	i.e. $\Fhat \subseteq E(G[V_1])$.
	Because no edges from $G[V_2]$ and $F$ are removed and $G[V_2]$ is connected,
	all edges in $G[V_2]$ and $F$ must not be cut and hence can be contracted into a single vertex $v_2$.
	
	So using the induction hypothesis on the cut $\Fhat$
	in $G/V_2/F$ with the terminal separation given by
	all of $S_2$ replaced by $v_2$ gives that $\Fhat$
	has at most $c-1$ edges in any connected component of
	\[
	\left(G/V_2/F \right) \bs F_1.
	\]
	Because connected components are unchanged under contracting
	connected subsets, we get that $\Fhat$ has at most $c-1$
	edges in any connected components of $G \bs F_1$ as well.
\end{proof}

Now, a natural way to handle the case where a minimum Steiner
cut has a single terminal $s$ on some side is to use Lemma \ref{lem:IsolateContract}
to contract across the cut to make progress.
However, it may be the case that for some $s \in S$,
there are are many minimum $s$-isolating cuts:
consider for example the length $n$ path with only the endpoints
as terminals. If we always pick the edge closest to $s$ as the minimum $s$-isolating
cut, we may have to continue $n$ rounds, and thus add all $n$ edges
to our set of intersecting edges.

To remedy this, we instead pick a ``maximal" $s$-isolating minimum cut.
One way to find a maximal $s$-isolating cut is to repeatedly contract
across an $s$-isolating minimum cut
using Lemma \ref{lem:IsolateContract} until its size increases.
At that point, we add the last set of edges found in the cut to
the set of intersecting edges.
We have made progress because the value of the minimum
$s$-isolating cut in the contracted graph must have increased by at
least $1$. While there are many ways to find a maximal $s$-isolating minimum cut,
the way described here extends to our analysis in Section \ref{sec:ExpanderLocalCut}.

Pseudocode of this algorithm is shown in 
Figure~\ref{fig:GetIntersectingEdgesSteiner}, and the procedure
for the repeated contractions to find a maximal $s$-isolating cut
described in the above paragraph is in Line~\ref{ln:IsolateContract}.

\begin{figure}[!h]
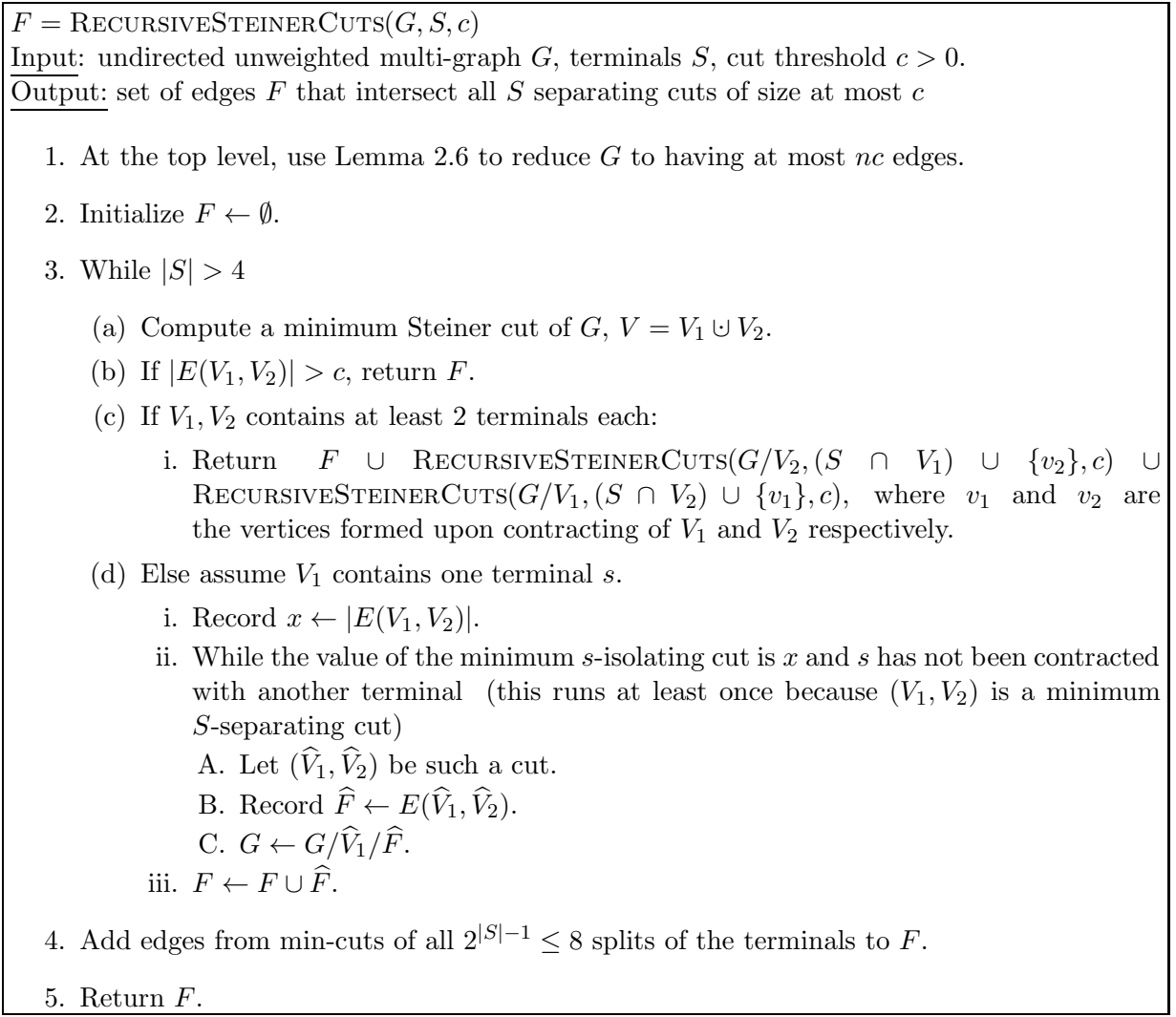

	\begin{algbox}
		$F = \textsc{RecursiveSteinerCuts}(G, S, c)$
		
		\underline{Input}: undirected unweighted multi-graph $G$,
		terminals $S$, cut threshold $c > 0$.
		
		\underline{Output:} set of edges $F$ that intersect all $S$
		separating cuts of size at most $c$
		
		\begin{enumerate}
			\item At the top level, use Lemma \ref{lem:EdgeReduction} to reduce $G$ to
			having at most $nc$ edges.
			\item Initialize $F \leftarrow \emptyset$.
			\item While $|S| > 4$
			\begin{enumerate}
				\item Compute a minimum Steiner cut of $G$,
					$V = V_1 \cupdot V_2$.
				\item If $|E(V_1, V_2)| > c$, return $F$.
				\item If $V_{1}, V_2$ contains at least $2$ terminals each:
				\begin{enumerate}
					\item Return $F \cup \textsc{RecursiveSteinerCuts}(G/V_2, (S \cap V_1) \cup \{v_2\}, c) \cup \textsc{RecursiveSteinerCuts}(G / V_1, (S \cap V_2) \cup \{v_1\}, c)$, where $v_1$ and $v_2$ are the vertices
					formed upon contracting of $V_1$ and $V_2$ respectively. \label{ln:Return}
				\end{enumerate}
				\item Else assume $V_1$ contains one terminal $s$. \label{ln:IsolateContract}
				\begin{enumerate}
					\item Record $x \leftarrow |E(V_1, V_2)|$.
					\item While the value of the minimum $s$-isolating cut is $x$ and $s$ has not been
						contracted with another terminal \label{ln:Condition}
					\label{ln:LoopToMaximal}
					(this runs at least once because $(V_1, V_2)$ is a minimum $S$-separating cut)
						\begin{enumerate}
							\item Let $(\Vhat_1, \Vhat_2)$ be such a cut.
							\item Record $\Fhat \leftarrow E(\Vhat_1, \Vhat_2)$.
							\item $G\leftarrow G/\Vhat_1/\Fhat.$ \label{ln:ContractedGraph}
						\end{enumerate}
					\item $F \leftarrow F \cup \Fhat$.
				\end{enumerate}
			\end{enumerate}
			\item Add edges from min-cuts of all $2^{|S|-1} \leq 8$ splits of the terminals to $F$.
			\item Return $F$.
		\end{enumerate}
	\end{algbox}
	\caption{Recursive algorithm using Steiner minimum cuts
		for finding a set of edges
		that intersect all terminal cuts of size $\leq c$.}
	\label{fig:GetIntersectingEdgesSteiner}
\end{figure}

\paragraph{Discussion of algorithm in Figure \ref{fig:GetIntersectingEdgesSteiner}.}
We clarify some lines in the algorithm of Figure \ref{fig:GetIntersectingEdgesSteiner}.
If the algorithm finds a nontrivial $S$-separating cut as the Steiner minimum cut, it returns the result of the recursion in Line~\ref{ln:Return},
and does not execute any of the later lines in the algorithm.
In Line~\ref{ln:Condition}, in addition to checking that the $s$-isolating minimum cut size is still $x$, we also must check that $s$ does not get contracted with another terminal.
Otherwise, contracting across that cut makes global progress by reducing the number of terminals by $1$.
In Line~\ref{ln:ContractedGraph}, note that we can still view $s$
as a terminal in $G\leftarrow G/\Vhat_1/\Fhat$,
as we have assumed that this contraction does not merge $s$
with any other terminals.

\begin{lemma}
\label{lem:WeakerVersion}
	For any graph $G$, terminals $S$, and cut value $c$,
	Algorithm \textsc{RecursiveSteinerCuts} as shown in
	Figure~\ref{fig:GetIntersectingEdgesSteiner} runs in $\O(n^2c^4)$ time
	and returns a set at most $O(|S|c^2)$ edges that
	intersect all $(S,c)$-cuts.
\end{lemma}

\begin{proof}
We assume $m \le nc$ throughout, as we can reduce to this case in $O(mc)$ time
by Lemma \ref{lem:EdgeReduction}.

Note that the recursion in Line~\ref{ln:Return} can only branch $O(|S|)$ times, 
by the analysis in Lemma \ref{lem:RecursiveNontrivialCutsCorrectness}. Similarly,
the case where $s$ gets contracted with another terminal in Line~\ref{ln:Condition}
can only occur $O(|S|)$ times.

Therefore, we only create $O(|S|)$ distinct terminals throughout the algorithm. Let $s$ be a terminal created
at some point during the algorithm. By monotonicity of cuts in Lemma \ref{lem:CutMonotone},
the minimum $s$-isolating cut can only increase in size $c$ times, hence $F$ is the union of $O(|S|c)$ cuts
of size at most $c$. Therefore, $F$ has at most $O(|S|c^2)$ edges.

To bound the runtime, we use the total number of edges in the graphs in our recursive algorithm
as a potential function. This potential function starts at $m$.
Note that the recursion of Line~\ref{ln:Return}
can increase the potential function by $c$,
hence the total potential function increase
throughout the algorithm is bounded by $m + O(c|S|) = O(nc)$.

Each loop of Line~\ref{ln:Condition} decreases our potential function
by at least $1$ from contractions.
Thus, the total runtime of the loop involving Line~\ref{ln:Condition}
can be bounded by
\[
O\left(mc\right) + \O\left(m+nc^3\right)
=
\O\left(nc^3\right),
\]
where the former term is from running a maxflow algorithm up to flow $c$,
and the latter is from an execution of the Steiner minimum cut algorithm
in~\cite{ColeH03}.
As the total potential function increase is at most $O(nc)$,
the loop in Line~\ref{ln:Condition} can only execute $O(nc)$ times,
for a total runtime of $\O(n^2c^4)$ as desired.
\end{proof}

Our further speedup of this routine in Section~\ref{sec:Expander}
also uses a faster variant of \textsc{RecursiveSteinerCuts} as
base case, which happens when $|S|$ is too small.
Here the main observation is that

A modification to Algorithm  as shown in
	Figure~\ref{fig:GetIntersectingEdgesSteiner} can reduce the runtime.

\begin{lemma}
\label{lem:StrongerVersion}
	For any graph $G$, terminals $S$, and cut value $c$,
	there is an algorithm that runs in $\O(mc + n|S|c^4)$ time
	and returns a set at most $O(|S|c^2)$ edges that
	intersect all $(S,c)$-cuts.
\end{lemma}
\begin{proof}
We modify \textsc{RecursiveSteinerCuts} as shown in
Figure~\ref{fig:GetIntersectingEdgesSteiner}
and its analysis as given in Lemma \ref{lem:WeakerVersion} above.
Specifically, we modify how we compute a maximal $s$-isolating minimum cut in Line~\ref{ln:Condition}.
For any partition $S = S_1 \cupdot S_2$, by submodularity of cuts
it is known that there is a unique maximal subset $V_1 \subseteq V$ such that
\begin{align*}
S_1 & \subseteq V_1,\\
S_2 & \subseteq V\bs V_1,\\
\left|E\left(V_1, V_2\right)\right|
&=
\left|\Mincut\left(G,S_1,S_2\right)\right|.
\end{align*}
Also, this maximal set can be computed in $O(mc)$ time by running
the Ford-Fulkerson augmenting path algorithm,
with $S_2$ as source, and $S_1$ as sink.
The connectivity value of $c$ means at most $c$ augmenting paths
need to be found, and the set $V_2$ can be set to the vertices that
can still reach the sink set $S_2$ in the residual graph~\cite{FulkersonH75}. Now set $V_1 = V\bs V_2.$
Thus by setting $S_1 \leftarrow \{s\}$,
we can use the corresponding computed
set $V_1$ as the representative of the maximal $s$-isolating
Steiner minimum cut.

Now we analyze the runtime of this procedure.
First, in we reduce the number of edges to at most $nc$
in $O(mc)$ time.
As in the proof of Lemma \ref{lem:WeakerVersion}, all graphs in the recursion have at most $O(nc)$ edges.
The recursion in Line~\ref{ln:Return} can only branch $|S|$ times,
and we only need to compute $O(c|S|)$ maximal
$s$-isolating Steiner minimum cuts throughout the algorithm.
Each call to the Cole-Hariharan algorithm~\cite{ColeH03} requires
$\O(m+nc^3) = \O(nc^3)$ time,
for a total runtime of $\O(nc^3 \cdot c|S|) = \O(n|S|c^4)$ as desired.
\end{proof}

\section{Nearly-Linear Time Constructions Using Expanders}
\label{sec:Expander}

In this section 
We now turn our attention to efficiently finding these vertex sparsifiers.
Here we utilize insights from recent results on finding $c$-vertex
cuts~\cite{NanongkaiSY19a,NanongkaiSY19b:arxiv,ForsterY19:arxiv},
namely that in a well connected graph, any cut of size at most $c$
must have a very small side.
This notion of connectivity is
formalized through the notion of graph conductance.

\begin{definition}
	In an undirected unweighted graph $G = (V, E)$, denote
	the volume of a subset of vertices, $\vol(S)$, as the
	total degrees of its vertices.
	The conductance of a cut $S$ is then
	\[
	\Phi_G\left( S \right)
	=
	\frac{\abs{\partial\left( S \right)}}
		{\min\left\{ \vol\left( S \right), \vol\left(V \bs S \right) \right\}},
	\]
	and the conductance of a graph $G = (V, E)$ is the minimum
	conductance of a subset of vertices:
	\[
	\Phi\left( G \right)
	= \min_{S \subseteq V} \Phi_{G}\left( S \right).
	\]
\end{definition}

The ability to remove edges and add terminals means we
can use expander decomposition to reduce to the case where
the graph has high conductance.
Here we utilize expander decompositions, as stated by
Saranurak and Wang~\cite{SaranurakW19}:

\begin{lemma}
\label{lem:ExpanderDecompose}
	(Theorem 1.2. of~\cite{SaranurakW19}, Version 2 \url{https://arxiv.org/pdf/1812.08958v2.pdf})
	There exists an algorithm \textsc{ExpanderDecompose}
	that for any undirected unweighted graph $G$ and
	any parameter $\phi$, decomposes in $O(m \log^{4}{n} \phi^{-1})$
	time $G$ into pieces of conductance at least $\phi$
	so that at most $O(m \phi \log^{3}n)$ edges are between
	the pieces.
\end{lemma}

Note that if a graph has conductance $\phi$,
any cut of size at most $c$ must have
\begin{equation}
\min\left\{ \vol\left( S \right), \vol\left(V \bs S \right) \right\}
\leq
c\phi^{-1}.
\label{eq:ExpanderImba}
\end{equation}

Algorithmically, we can further leverage it in two ways,
both of which are directly motivated by recent works
on vertex connectivity~\cite{NanongkaiSY19a,ForsterY19:arxiv,NanongkaiSY19b:arxiv}.

\subsection{Enumeration of All Small Cuts by their Smaller Sides}

In a graph with expansion $\phi$, we can enumerate all cuts of size
at most $c$ in time exponential in $c$ and $\phi$.
\begin{lemma}
\label{lem:EnumerateCuts}
	In a graph $G$ with conductance $\phi$ we can enumerate all cuts of size at most $c$
	with connected smaller side in time $O(n \cdot (c\phi^{-1})^{2c})$.
\end{lemma}

\begin{proof}
	We first enumerate over all starting vertices.
	For a starting vertex $u$,
	we repeatedly perform the following process.
	\begin{enumerate}
		\item perform a DFS from $u$ until it reaches more than
		$c \phi^{-1}$ vertices.
		\item Pick one of the edges among the reached vertices as
		a cut edge.
		\item Remove that edge, and recursively start another DFS
		starting at $u$.
	\end{enumerate}
	After we have done this process at most $c$ times, we check whether
	the edges form a valid cut, and store it if so.

	By Equation~\ref{eq:ExpanderImba}, the smaller side of the
	can involve at most $c \phi^{-1}$ vertices.
	Consider such a cut with $S$ as the smaller side,
	$F = E(S, V \setminus S)$, and $|S| \leq c \phi^{-1}$.
	Then if we picked some vertex $u \in S$ as the starting point,
	the DFS tree rooted at $u$ must contain some edge in $F$
	at some point.
	Performing an induction with this edge removed then gives
	that the DFS starting from $u$ will find this cut.
	
	Because there can be at most $O((c \phi^{-1})^2)$ different
	edges picked among the vertices reached, the total work performed
	in the $c$ layers of recursion is $O((c \phi^{-1})^{2c})$.
\end{proof} 

Furthermore, it suffices to enumerate all such cuts once at the start,
and reuse them as we perform contractions.
\begin{lemma}
	\label{lem:ContractionCut}
	If $\Fhat$ is a set of edges that form a cut in $G / \Ehat$,
	that is, $G$ with a subset of edges $\Ehat$ contracted,
	then $\Fhat$ is also a cut in $G$.
\end{lemma}
Note that this lemma also implies that an expander stays
so under contractions.
So we do not even need to re-partition the graph as we recurse.

\begin{proof}[Proof of Theorem~\ref{thm:IntersectingEdges}
	Part~\ref{part:Slower}]
	First, we perform expander decomposition, remove the inter-cluster edges, and add their endpoints as terminals.

	Now, we describe the modifications to
	\textsc{GetIntersectingEdgesSlow} that makes it efficient.
	
	Lemma~\ref{lem:AddEdges}, and Lemma~\ref{lem:Independent}
	allows us to consider the pieces separately.

	Now at the start of each recursive call, enumerate
	all cuts of size at most $c$, and store the vertices
	on the smaller side, which by
	Equation~\ref{eq:ExpanderImba} above has size at most
	$O(c \phi^{-1})$.
	When such a cut is found, we only invoke recursion on the
	smaller	side (in terms of volume).
	For the larger piece, we can continue using the original set
	of cuts found during the search.
	
	To use a cut from a pre-contracted state, we need to:
	\begin{enumerate}
		\item check if all of its edges remain (using a union-find data structure).
		\item check if both portions of the graph remain
		connected upon removal of this cut -- this can be done
		by explicitly checking the smaller side, and certifying
		the bigger side using a dynamic connectivity data structure by
		removing all edges from the smaller side.
	\end{enumerate}
	Since we contract each edge at most once, the total work
	done over all the larger side is at most
	\[ \O\left( m \left( c\phi \right)^{-2c} \right), \] where we
	have included the logarithmic factors from using the dynamic connectivity
	data structure.
	Furthermore, the fact that we only recurse on things with half
	as many edges ensures that each edge participates in the cut
	enumeration process at most $O(\log n)$ times.
	Combining these then gives the overall running time.
\end{proof}

\subsection{Using Local Cut Algorithms}
\label{sec:ExpanderLocalCut}
A more recent development are local cut algorithms, which
for a vertex $v$ can whether there is a cut of size at most $c$
such that the side with $v$ has volume at most $\nu$. The runtime
is linear in $c$ and $\nu$.
\begin{theorem}[Theorem 3.1 of \cite{NanongkaiSY19b:arxiv}]
\label{thm:LocalFlow}
Let $G$ be a graph and let $v \in V(G)$ be a vertex. For a connectivity parameter $c$
and volume parameter $\nu$, there is an algorithm that with high probability
either
\begin{enumerate}
\item Certifies that there is no cut of size at most $c$ such that the side with $v$
has volume at most $\nu$.
\item Returns a cut of size at most $c$ such that the side with $v$ has volume
at most $130c\nu.$ It runs in time $\O(c^2\nu).$
\end{enumerate}
\end{theorem}
Let $v$ be a vertex. We now formalize the notion of the smallest cut
that is \emph{local} around $v$.
\begin{definition}[Local cuts]
For a vertex $v \in G$ define $\localcut(v)$ to be
\[ \min_{\substack{V = V_1 \cupdot V_2 \\ v \in V_1 \\ \vol(V_1) \le \vol(V_2)}} |E(V_1,V_2)|.\]
\end{definition}
We now combine Theorem \ref{thm:LocalFlow} with the observation from Equation~\ref{eq:ExpanderImba}
in order to control the volume of the smaller side of the cut in an expander.
\begin{lemma}
	\label{lem:ExpanderCutFaster}
Let $G$ be a graph with conductance at most $\phi$, and let $S$ be a set of terminals.
If $|S| \ge 500c^2\phi^{-1}$ then for any vertex $s \in S$ we can with high probability
in $\O(c^3\phi^{-1})$ time either compute $\localcut(s)$ or certify that $\localcut(s) > c.$
\end{lemma}
\begin{proof}
We binary search on the size of the minimum Steiner cut with $s$
on the smaller side, and apply Theorem \ref{thm:LocalFlow}.
The smaller side of a Steiner cut has volume at most $c\phi^{-1}.$ Therefore,
if $|S| \ge 500c^2\phi^{-1}$ then the cut returned by Theorem \ref{thm:LocalFlow}
for $\nu = c\phi^{-1}$ will always be a Steiner cut, as $130\nu c \le |S|/2.$
The runtime is $\O(\nu c^2) = \O(c^3 \phi^{-1})$ as desired.
\end{proof}
We can substitute this faster cut-finding procedure
into \textsc{RecursiveSteinerCuts} to get the faster
running time stated in Theorem~\ref{thm:IntersectingEdges}
Part~\ref{part:Faster}.
\begin{proof}[Proof of Theorem~\ref{thm:IntersectingEdges}
	Part~\ref{part:Faster}]
First, we perform expander decomposition, remove the inter-cluster edges, and add their endpoints as terminals.

Now, we describe the modifications we need to make to Algorithm \textsc{RecursiveSteinerCuts} as shown in
	Figure~\ref{fig:GetIntersectingEdgesSteiner}.
First, we terminate if $|S| \le 500c^2\phi^{-1}$ and use the result of Lemma \ref{lem:StrongerVersion}.
Otherwise, instead of using the Cole-Hariharan algorithm, we compute the terminal $s \in S$ with minimal value of $\localcut(s).$
This gives us a Steiner minimum cut. If the corresponding cut
is a nontrivial $S$-separating cut then we recurse as in Line~\ref{ln:Return}. Otherwise, we perform
the loop in Line~\ref{ln:Condition}.

We now give implementation details for computing the terminal $s \in S$ with minimal value of $\localcut(s)$.
By Lemma \ref{lem:CutMonotone} we can see that for a terminal $s$, $\localcut(s)$ is monotone
throughout the algorithm. For each terminal $s$, our algorithm records the previous value of $\localcut(s)$ computed.
Because this value is monotone, we need only check vertices $s$ whose value of $\localcut(s)$ could still possibly be minimal.
Now, either $\localcut(s)$ is certified to be minimal among all $s$, or the value of $\localcut(s)$ is higher than
the previously recorded value. Note that this can only occur $O(c|S|)$ times, as we stop processing a vertex $s$ if $\localcut(s) > c$.

We now analyze the runtime. We first bound the runtime from the cases $|S| \le 500c^2\phi^{-1}$. The total number of vertices and edges
in the leaves of the recursion tree is at most $O(mc).$ Therefore, by Lemma \ref{lem:StrongerVersion}, the total runtime from these is
at most \[ \O(500c^2\phi^{-1} \cdot mc \cdot c^4) = \O(m\phi^{-1}c^7). \]

Now, the loop of Line~\ref{ln:Condition} can only execute $c\phi^{-1}$ times, because the volume of any $s$-isolating cut has size at most
$c\phi^{-1}.$ Each iteration of the loop requires $\O(c^3\phi^{-1})$ time by Lemma \ref{lem:ExpanderCutFaster}. Therefore, the total runtime of executing the loop and calls to it is bounded by
\[ \O\left(c|S| \cdot c\phi^{-1} \cdot c^3\phi^{-1} \right) = \O(|S|\phi^{-2}c^6). \] Combining these shows Theorem~\ref{thm:IntersectingEdges}
	Part~\ref{part:Faster}.
\end{proof}

\section*{Acknowledgements}

We thank Gramoz Goranci, Jakub Lacki, Thatchaphol Saranurak,
and Xiaorui Sun for multiple enlightening discussions on this topic.
In particular, the graph partitioning based efficient constructions in
Section~\ref{sec:Expander} are directly motivated by conversations
with Xiaorui.
We also thank Paraniya Chamelsook, Syamantak Das, Bundit Laekhanukit,
Antonio Molina,  Bryce Sandlund, and Daniel Vaz
for communicating with us about unpublished related results.

\bibliographystyle{alpha}

\input{VertexSparsifierCEdge.bbl}


\end{document}

%% file: VertexSparsifierCEdge.bbl
\newcommand{\etalchar}[1]{$^{#1}$}